\documentclass[10pt,english]{article}
\usepackage{a4wide}
\usepackage{microtype}
\usepackage{amsmath}
\usepackage{mathrsfs}

\usepackage{tikz}
\usepackage{subcaption}

\pgfdeclarelayer{bg}    
\pgfsetlayers{bg,main}

\newtheorem{theorem}{Theorem}
\newtheorem{lemma}{Lemma}

\newenvironment{proof}[1][Proof]{\begin{trivlist}
\item[\hskip \labelsep {\bfseries #1}]}{\end{trivlist}}

\title{Counting Shortest Two Disjoint Paths in Cubic Planar Graphs with an NC Algorithm}
\author{Andreas Bj\"orklund and Thore Husfeldt}
\date{}
\begin{document}
\maketitle
\begin{abstract}
Given an undirected graph and two disjoint vertex pairs $s_1,t_1$ and $s_2,t_2$,  the Shortest two disjoint paths problem (S2DP)
asks for the minimum total length of two vertex disjoint paths connecting $s_1$ with $t_1$, and $s_2$ with $t_2$, respectively.

We show that for cubic planar graphs there are NC algorithms, uniform circuits of polynomial size and polylogarithmic depth, that compute the S2DP and moreover also output the number of such minimum length path pairs.

Previously, to the best of our knowledge, no deterministic polynomial time algorithm was known for S2DP in cubic planar graphs with arbitrary placement of the terminals. In contrast, the randomized polynomial time algorithm by Bj\"orklund and Husfeldt, ICALP 2014, for general graphs is much slower, is serial in nature, and cannot count the solutions.

Our results are built on an approach by Hirai and Namba, Algorithmica 2017, for a generalisation of S2DP, and fast algorithms for counting perfect matchings in planar graphs.
\end{abstract}

\section{Introduction}
\emph{Shortest disjoint $A,B$-paths}, introduced by Hirai and Namba~\cite{HN17}, is the following problem: Let $G=(V,E)$ be an undirected graph with two non-empty disjoint vertex subsets $A,B\subseteq V$ of even size and an edge length function $\ell\colon E\rightarrow \{1,\ldots,L\}$.
An edge subset $E'\subseteq E$ is a solution to \emph {Disjoint $A,B$-paths} if it consists of $\frac{1}{2}(|A|+|B|)$ disjoint paths with endpoints both in $A$ or both in $B$. 
The length $\ell(E')$ of a solution is $\sum_{e\in E'} \ell(e)$, and a \emph{shortest} solution has length $\ell_{A,B}=\operatorname{min}_{E'}\ell(E')$.
The objective is to compute $\ell_{A,B}$. 
The special case $|A|=|B|=2$ is a well-studied problem called \emph{Shortest two disjoint paths}.

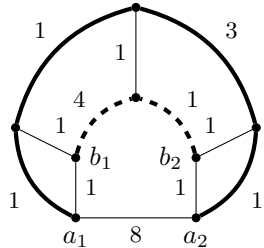
\begin{figure}[!hb]
\centering\begin{tikzpicture}[scale=.4, auto]
    \tikzstyle{node_style} = [circle,draw=black,fill=black, inner sep=0pt, minimum size=3pt]
    \tikzstyle{edge_style} = [draw=black, line width=2, thin]
    \tikzstyle{es1_style} = [draw=black, line width=2, ultra thick]
    \tikzstyle{es2_style} = [draw=gray, dashed, line width=2, ultra thick]

    \node[node_style] (v1) at (0,4) {};
    \node[node_style] (v2) at (0,1) {};
    \node[node_style] (v3) at (-4,0) {};
    \node[node_style] (v4) at (4,0) {};
    \node[node_style, label=right:$b_1$] (v5) at (-2,-1) {};
    \node[node_style, label=left:$b_2$] (v6) at (2,-1) {};
    \node[node_style, label=below:$a_1$] (v7) at (-2,-3) {};
    \node[node_style, label=below:$a_2$] (v8) at (2,-3) {}; 
    
    \begin{pgfonlayer}{bg} 
    \draw[edge_style]  (v1) edge node[swap] {\small $1$} (v2);
    \draw[es1_style]    (v1) edge [bend right]  node[swap] {\small $1$} (v3);
    \draw[es1_style]    (v1) edge [bend left] node {\small $3$} (v4);
    \draw[edge_style]  (v3) edge node {\small $1$} (v5);
    \draw[es2_style]    (v5) edge [bend left] node {\small $4$} (v2);
    
    \draw[edge_style]  (v4) edge node[swap] {\small $1$} (v6) ;
    \draw[es2_style]    (v6) edge [bend right] node[swap] {\small $1$} (v2);
    \draw[es1_style]    (v3) edge [bend right] node[swap] {\small $1$} (v7); 
    \draw[es1_style]  (v4) edge [bend left] node {\small $1$} (v8);

    \draw[edge_style]  (v5) edge node {\small $1$} (v7);
    \draw[edge_style]  (v6) edge node[swap] {\small $1$} (v8);
    \draw[edge_style]  (v7) edge node[swap] {\small $8$} (v8); 
    \end{pgfonlayer}
    \end{tikzpicture}
    \caption{A solution of minimum length $\ell_{A,B}=11$ to Shortest disjoint $A,B$-paths with $A=\{a_1,a_2\}$ and $B=\{b_1, b_2\}$.
  Since $|A|=|B|=2$, this is also an example of Shortest two disjoint paths. Note that neither path is a shortest path between its terminals.}
    \label{fig: Ex}
\end{figure}

We write $S_{A,B}$ for the number of solutions of length $\ell_{A,B}$.
A graph is \emph{cubic} (sometimes called 3-regular) if every vertex has degree 3.
We prove the following:

\begin{theorem}
\label{thm: alg}
For any cubic planar $n$-vertex graph $G=(V,E)$, disjoint vertex subsets $A$ and $B$, and edge length function $\ell:E\rightarrow \{1,\ldots, L\}$, we can compute $\ell_{A,B}$ and $S_{A,B}$ in deterministic $\tilde{O}(2^{|A\cup B|}n^{\omega/2+2}L^2)$ time\footnote{The $\tilde{O}(f(n))$ notation suppresses factors polylogarithmic in $f(n)$.}, where $\omega<2.373$ is the exponent of square matrix multiplication.
\end{theorem}

In particular, for $|A|+|B|=O(1)$, the algorithm runs in deterministic time $\tilde{O}(n^{\omega/2+2}L^2)$.

To the best of our knowledge, no polynomial-time deterministic algorithm was known even for $|A|=|B|=2$.
Hirai and Namba's algorithm \cite{HN17} works for general graphs in randomized time $n^{O(|A\cup B|)}$, so Thm.~\ref{thm: alg} also shows that cubic planar graphs allow better exponential dependency on $|A\cup B|$. 

We focus on the algorithmically interesting cubic case; Sec.~\ref{sec: max degree 3} shows that all our algorithms extend to the case where the graph has maximum degree $3$.

\medskip
Because we can count the solutions we can use well-known techniques to retrieve a witness for the shortest length.
By using our algorithm as a subroutine, we can retrieve the $i$th witness in a lexicographical order of the solutions by a polynomial overhead self reduction, by peeling off edges one at a time and remeasuring the number of solutions. In particular, by choosing $i$ uniformly from $\{1,\ldots, S_{A,B}\}$, we can sample uniformly over the solutions without first explicitly constructing the list of solutions.

\medskip
Our algorithm is based on counting perfect matchings in a planar graph. Vazirani~\cite{V89} showed how every bit in the number of perfect matchings in planar graphs can be decided by an NC algorithm, i.e., uniform polylogarithmically shallow polynomial size circuits, an observation he attributes to Mike Luby. By using his algorithm as a subroutine, we present an efficient parallel algorithm, which we state here for the special case of  Shortest two disjoint paths:

\begin{theorem}
\label{thm: nc}
For any cubic planar $n$-vertex graph $G=(V,E)$, disjoint vertex subsets $A$ and $B$ with $|A|=|B|=2$, and edge length function $\ell:E\rightarrow \{1,\ldots, L\}$, we can compute $\ell_{A,B}$ and $S_{A,B}$ by an NC algorithm.
\end{theorem}

The same statement holds as long as $|A|+|B|$ is logarithmic in $n$.

\medskip
Via the Isolation lemma of Mulmuley, Vazirani, and Vazirani~\cite{MVV87} we can also obtain a witness, i.e., a solution $E'$ of length $\ell(E') = \ell_{A,B}$, with a \emph{randomized} NC algorithm. We note that the recent breakthrough result showing how to find a perfect matching in a planar graph in NC by Anari and Vazirani~\cite{AV17} doesn't seem to be directly applicable to our problem. Our algorithm counts the solutions to Shortest two disjoint paths by an annihilation sieve, i.e. the number of solutions is an alternating sum of perfect matchings in a set of graphs, but many of the terms will cancel each other. Hence there are many perfect matchings that do not correspond to a solution. Finding one unconditionally won't help us.

\medskip
We also provide some evidence that the exponential dependence on $|A|+|B|$ in the running time is probably necessary:

\begin{theorem}
\label{thm: nph}
For any cubic planar graph and two disjoint vertex subsets $A$ and $B$, it is \#P-hard to simultaneously compute the length and the number of solutions to Shortest disjoint $A,B$-paths.
\end{theorem}

\subsection{Hirai and Namba's Result}
Hirai and Namba~\cite{HN17} shows that Shortest disjoint $A,B$-paths has a randomized algorithm running in $n^{O(|A\cup B|)}$ time, that w.h.p. finds the length of the shortest disjoint paths. Their algorithm is inspired by the algorithm of Gallai~\cite{G61} that can be used to address the special case $B=\emptyset$, and the algorithm by Bj\"orklund and Husfeldt~\cite{BH14} for the special case $|A|=|B|=2$.
They apply a two-step method. First it expands the input graph $G$ into another edge weighted graph $G'$ using so-called Gallai paths, and argues that the (weighted) perfect matchings in $G'$ can be used to obtain the solution to the original problem. Second, it uses the fact that counting perfect matchings in $G'$ modulo $2^k$ has a $n^{O(k)}$ time algorithm. By design of their reduction, the solutions will be counted $2^{|A\cup B|/2}$ times each, so they need to set $k>\frac{1}{2}|A\cup B|$ to count something meaningful, but still set it small to keep the running time down. This means the algorithm only is capable of counting the solutions modulo a fixed small power of two. They need to use the Isolation lemma~\cite{MVV87} to make sure the solutions count a number of times that is not divisible by this small power of two. This is why they need randomness, and the same problem preventing a deterministic algorithm is present in the earlier Shortest two disjoint paths algorithm by Bj\"orklund and Husfeldt~\cite{BH14}.

\subsection{Our Approach}
We apply Hirai and Namba's approach to the planar cubic case. It is well-known that in any planar graph we can count the perfect matchings in polynomial time. In particular we don't just obtain the result modulo a small power of two. We will use this, but there is one obstacle that needs to be addressed to accomplish this: The reduction Hirai and Namba use does not preserve planarity.
Our contribution is to show that for cubic planar graphs, we can construct a set of $2^{|A\cup B|/2}$ cubic planar graphs, each having non-negative edge weights, so that a linear combination of the number of weighted perfect matchings in these graphs can be used to deduce the number of solutions to the Shortest disjoint $A,B$-paths in the original instance. We are inspired by the result of Galluccio and Loebl~\cite{GL99} that shows how to count perfect matchings in graphs of genus $g$ by constructing $4^g$ orientations and computing the Pfaffian for each of them. We choose to use a more direct approach instead of reducing to their result to make the description of our algorithm more self contained.  

\subsection{Related Work}
Bj\"orklund and Husfeldt~\cite{BH14} showed that Shortest two disjoint paths in a general unweighted undirected graph has a polynomial time Monte Carlo algorithm.
Colin de Verdi\'ere and Schrijver~\cite{CS11}, and Kobayachi and Sommer~\cite{KS10} showed that for planar graphs, deterministic polynomial-time algorithms for the Shortest two disjoint paths exist if the four terminals lie on the boundary of at most two faces. Our new algorithm works for all cubic planar cases, but is much slower. Still it is significantly faster than the general $O(n^{11})$ time algorithm by Bj\"orklund and Husfeldt~\cite{BH14}. 

Very recently, Datta et al.~\cite{DIKM18} presented a deterministic algorithm independently of ours for Shortest $k$-disjoint paths in planar graphs conditioned on the terminals either all being placed on the same face or all source terminals on one face and the target terminals on another. Interestingly, their algorithm is also based on computing determinants just as ours and is capable of counting the solutions just as our algorithms can, although they don't use Pfaffian orientations as we do. In particular, restricted to the Shortest two disjoint path problem, their algorithm does not solve the case when the four terminals are not all incident to at most two faces.

If we only want to decide if two disjoint paths joining given vertex pairs exist, no matter their length, deterministic polynomial-time algorithms have been known since 1980 for general graphs, by Ohtsuki~\cite{O80}, Seymour~\cite{Se80}, Shiloah~\cite{Sh80}, and Thomassen~\cite{T80}; all published independently. Tholey~\cite{T06} reduced the running time for that problem to near-linear. Khuller, Mitchell, and Vazirani~\cite{KMV92} showed that the problem can be solved in NC.

The $k$-disjoint paths problem is the natural generalisation of the two disjoint paths problem:
Given a list $\{(s_1,t_1), \dots, (s_k, t_k)\}$ of terminal pairs,  decide if there exist $k$ disjoint paths connecting $s_i$ with $t_i$ for $i\in \{1,\ldots,k\}$.
Again neglecting the length of the solution, this problem has a
polynomial time algorithm in general graphs for fixed $k$, but the dependence on $k$ is horrible ($\operatorname{exp}\operatorname{exp}\operatorname{exp}\operatorname{exp}O(k)$, see~\cite{KW10}). For planar graphs, there exists a doubly exponential
$(\operatorname{exp}\operatorname{exp}O(k))\operatorname{poly}(n)$ time algorithm, by Adler et al. \cite{AKKLST11}. 
 For comparison, our running time dependence is singly exponential in the size of the terminal set, but of course our criteria for allowed connections is much relaxed.
 
The special case $B=\emptyset$ in Disjoint $A,B$-paths is referred to as $A$-Paths in Lov\'asz and Plummer~\cite{LP86}. Its solution in general undirected graphs by a polynomial time algorithm was given by Gallai~\cite{G61} by a reduction to finding a perfect matching. Using Mulmuley, Vazirani, and Vazirani's algorithm for the problem they call Exact Matching~\cite{MVV87} on Gallai's construction, one can in randomized polynomial time solve the Shortest disjoint $A$-paths.

The idea of using fast perfect matching counting in restricted graph classes to solve other combinatorial optimisation problems is not new, a prominent example is the polynomial time algorithm for Maximum Cut in an unweighted graph of bounded genus by Galluccio, Loebl, and Vondr\' ak \cite{GLV01}.

\section{Algorithmic Results}
In this section we will prove Thm.~\ref{thm: alg} and~\ref{thm: nc}.

\subsection{Notation}

A $(u,v)$-path is a path from vertex $u$ to vertex $v$.

A \emph{perfect matching} in an undirected graph $G=(V,E)$, is a subset $E'\subseteq E$ of the edges of size $|E'|=\frac{1}{2}|V|$, such that every vertex in $v\in V$ is the endpoint of exactly one edge in $E'$. Let $w\colon E\rightarrow \mathbf{N}$ be an edge weight function to positive integers. Let $\mathscr{M}(G)$ be the family of perfect matchings in $G$.
We denote by $\operatorname{pm}(G)$ the sum of the weighted perfect matchings in a graph $G$, i.e.,
\[
\operatorname{pm}(G)=\sum_{M\in \mathscr{M}(G)} \prod_{e\in M} w(e).
\]
If the weights are unity, this is the number of perfect matchings.
In our algorithm's analysis, some edges will be weighted by an indeterminate $s$ and $\operatorname{pm}(G)$ will be a polynomial in $s$. However, the algorithm itself will only work directly over the integers after replacing the indeterminate $s$ for a numerical value.

\subsection{Pfaffian Orientations}
A \emph{Pfaffian orientation} of a graph $G=(V,E)$ with edge weights $w\colon E\rightarrow \mathbf{N}$, is an orientation of the edges $q\colon E\rightarrow \{-1,1\}$ so that the skew-symmetric adjacency matrix $A_G$, where 
\begin{equation*}
\forall uv\in E,u<v\colon  q(u,v)w(uv)=A_G(u,v)=-A_G(v,u)=-q(u,v)w(uv),
\end{equation*}
satisfies
\begin{equation}
\label{eq: pfaff}
\operatorname{pm}(G)=\left|\sqrt{\det(A_G)}\right|.
\end{equation}
An orientation of a graph $G$ is Pfaffian if and only if every even-length cycle $C$ such that $G\setminus V(C)$ has a perfect matching, has an odd number of edges directed in either direction along $C$.
Kasteleyn~\cite{K57}, famously proved that all planar graphs have a Pfaffian orientation, and moreover showed how you given a planar graph can find a Pfaffian orientation fast. Nowadays it is even known how to find one in planar graphs in linear time, and Vazirani~\cite{V89} showed it can be computed in NC. 
In general it only holds that $\left|\operatorname{pm}(G)\right|=\left|\sqrt{\det(A_G)}\right|$, but we will only consider positive edge weights in this paper and hence already know $\operatorname{pm}(G)$ to be non-negative.
Little~\cite{L74} extended Kasteleyn's method to also work constructively for graphs that do not have a $K_{3,3}$ subgraph as a minor.
However, cubic $K_{3,3}$ minor free graphs coincide with the set of cubic planar graphs.

\subsection{Reduction from Disjoint A,B-Paths to Counting Perfect Matchings}

Consider as input a cubic planar graph $G$  and let $\ell\colon E\rightarrow \{1,\ldots,L\}$ be an edge length function, along with two disjoint subsets $A$ and $B$ of the vertices, each having even size. Set $\Lambda=\sum_{e\in E} \ell(e)$. We will reduce Shortest disjoint $A,B$-paths to counting perfect matchings so that planarity is preserved.
In this section, we will write $Z$ for the set of terminals, $Z=A\cup B$.

We first construct a larger graph $H$ from $G$ as follows.

Replace each nonterminal $v\in V\setminus Z$ with three vertices $h_1(v)$, $h_2(v)$, and $h_3(v)$ forming a triangle:
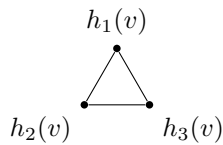
\[ \begin{tikzpicture}[scale =.5]
    \node (v_1) [circle, fill, inner sep =1pt, label= 90:{$h_1(v)$}] at (90:1cm) {};
    \node (v_2) [circle, fill, inner sep =1pt, label=210:{$h_2(v)$}] at (210:1cm) {};
    \node (v_3) [circle, fill, inner sep =1pt, label=330:{$h_3(v)$}] at (330:1cm) {};
    \draw (v_1)--(v_2);
    \draw (v_2)--(v_3);
    \draw (v_3)--(v_1);
\end{tikzpicture}
\]
Replace each terminal $z\in Z$ by a $3$-star on vertices $h_1(z)$, $h_2(z)$, $h_3(z)$, and \emph{terminal center} $h(z)$: 
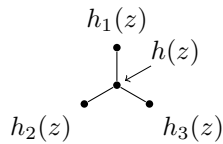
\[ \begin{tikzpicture}[scale =.5, every pin/.style={pin edge = {<-,black, shorten <= 2pt}, inner sep =0pt} ]
    \node (z_1) [circle, fill, inner sep =1pt, label= 90:{$h_1(z)$}] at (90:1cm) {};
    \node (z_2) [circle, fill, inner sep =1pt, label=210:{$h_2(z)$}] at (210:1cm) {};
    \node (z_3) [circle, fill, inner sep =1pt, label=330:{$h_3(z)$}] at (330:1cm) {};
    \node (z)   [circle, fill, inner sep =1pt,  pin=30:{$h(z)$} ] at (0,0) {};
    \draw (z)--(z_2);
    \draw (z)--(z_3);
    \draw (z)--(z_1);
\end{tikzpicture}
\]
We call the edges within these two gadgets \emph{internal} edges.

Moreover, if $uv\in E(G)$, then $h_i(u)h_j(v)$ is also an edge in $H$ for some $i,j\in\{1,2,3\}$ in such a way that each vertex in $H$ is used in exactly one of the additional edges. We call these edges in $H$ between gadgets \emph{external} edges. 

We write $f(uv)=h_i(u)$ and $g(uv)=h_j(v)$ to identify the two gadget vertices in $H$ connected by the external edge representing $uv$. 
Confer figure~\ref{fig: G'}. 
The graph $H$ has the property that every vertex except the terminal centers $h(z)$ for $z\in Z$ is part of exactly one external edge.

Our first insight is the following:

\begin{lemma}
If $G$ is planar, then so is $H$.
\end{lemma}
\begin{proof}
Both gadgets are easily seen to be planar. To see that $H$ is planar, use an embedding of $G$.
For each vertex $v$ in $G$, consider a small enough circle $C_v$ around $v$ containing no other edge or vertex. Now replace $v$ with a copy of its gadget small enough to fit $C_v$.
\end{proof}

Hence, if $G$ is planar, we can find a Pfaffian orientation of $H$, as well as for any subgraph of it, as any subgraph is also planar. (We note in passing that a Pfaffian orientation of a graph is not necessarily a Pfaffian orientation of its subgraph.)

From $H$, we will create several graphs depending on a subset of the terminal vertices. We write $H(X)$ for $X\subseteq Z$ to mean the graph obtained from $H$ by removing the terminal centers $h(z)$ and all incident edges for each terminal $z\notin X$. 
We have $H=H(Z)$.

\medskip
We introduce an indeterminate $s$ to control the length of the paths. We write $D(X,s)$ for a skew-symmetric adjacency matrix of a Pfaffian orientation of $H(X)$, where we have multiplied all entries representing an external edge $e$ in $H(X)$ with $s^{\ell(e)}$. 
 
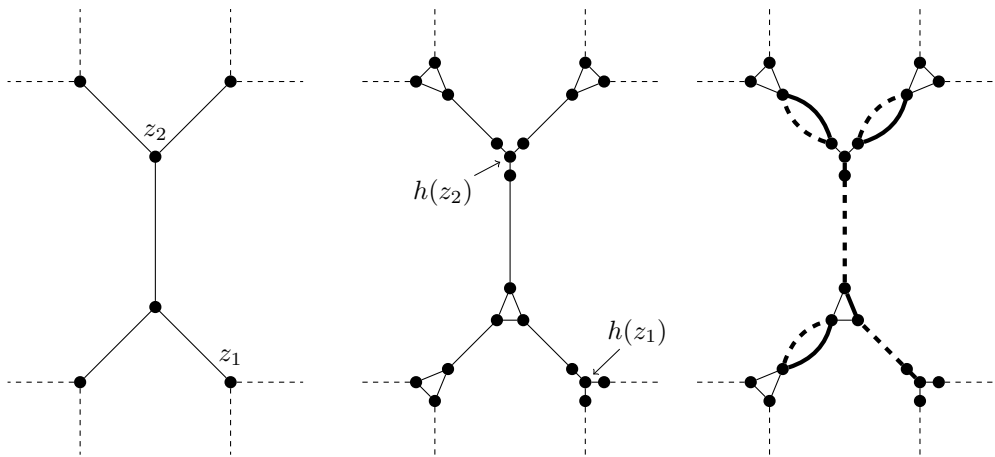
\begin{figure}[!ht]
\[\begin{tikzpicture}[scale=.5, rotate=90,shorten >=1pt, auto, node distance=1cm, ultra thick]
    \tikzstyle{node_style} = [circle,draw=black,fill=black, inner sep=0pt, minimum size=3pt]
    \tikzstyle{nhidden_style} = [inner sep=0pt, minimum size=0pt]
    \tikzstyle{edge_style} = [draw=black, line width=2, thin]
    \tikzstyle{ehidden_style} = [draw=black, line width=2, dotted, thin]
    
    \node[node_style,label=above:$z_1$] (v1) at (-4,-2) {};
    \node[node_style] (v2) at (-4,2) {};
    \node[node_style] (v3) at (-2,0) {};
    \node[node_style,label=above:$z_2$] (v4) at (2,0) {};
    \node[node_style] (v5) at (4,2) {};
    \node[node_style] (v6) at (4,-2) {};
    \node[nhidden_style] (v7) at (-6,-2) {};
    \node[nhidden_style] (v8) at (-6,2) {};
    \node[nhidden_style] (v9) at (-4,-4) {};
    \node[nhidden_style] (v10) at (-4,4) {};
    \node[nhidden_style] (v11) at (6,-2) {};
    \node[nhidden_style] (v12) at (6,2) {};
    \node[nhidden_style] (v13) at (4,-4) {};
    \node[nhidden_style] (v14) at (4,4) {};
 
    \begin{pgfonlayer}{bg} 
    \draw[edge_style]  (v1) edge (v3.center);
    \draw[edge_style]  (v2) edge (v3.center);
    \draw[edge_style]  (v3) edge (v4.center);
    \draw[edge_style]  (v4) edge (v5.center);
    \draw[edge_style]  (v4) edge (v6.center);
    
    \draw[ehidden_style]  (v1) edge (v7.center);
    \draw[ehidden_style]  (v1) edge (v9.center);
    \draw[ehidden_style]  (v2) edge (v8.center); 
    \draw[ehidden_style]  (v2) edge (v10.center);

    \draw[ehidden_style]  (v6) edge (v11.center);
    \draw[ehidden_style]  (v6) edge (v13.center);
    \draw[ehidden_style]  (v5) edge (v12.center); 
    \draw[ehidden_style]  (v5) edge (v14.center);
    \end{pgfonlayer}
    \end{tikzpicture}
    \qquad
    \begin{tikzpicture}[scale=.5,rotate=90, shorten >=1pt, auto, node distance=1cm, ultra thick,  every pin/.style={pin edge = {<-,black, shorten <= 2pt}, inner sep =0pt}]
    \tikzstyle{node_style} = [circle,draw=black,fill=black, inner sep=0pt, minimum size=3pt]
    \tikzstyle{nhidden_style} = [inner sep=0pt, minimum size=0pt]
    \tikzstyle{edge_style} = [draw=black, line width=2, thin]
    \tikzstyle{ehidden_style} = [draw=black, line width=2, dotted, thin]

    \node[node_style] (v1) [pin = 60:{$h(z_1)$}]at (-4,-2) {};
    \node[node_style] (v1e1) at (-4+.35,-2+.35) {};
    \node[node_style] (v1e2) at (-4-.5,-2) {};
    \node[node_style] (v1e3) at (-4,-2-.5) {};
    
    \node[node_style] (v2e1) at (-4+.35,2-.35) {};
    \node[node_style] (v2e2) at (-4-.5,2) {}; 
    \node[node_style] (v2e3) at (-4,2+.5) {};
    
    \node[node_style] (v3e1) at (-2-.35,0-.35) {};
    \node[node_style] (v3e2) at (-2-.35,0+.35) {};
    \node[node_style] (v3e3) at (-2+.5,0) {};

    \node[node_style] (v4) [pin = 210:{$h(z_2)$}] at (2,0) {};
    \node[node_style] (v4e1) at (2+.35,0-.35) {};
    \node[node_style] (v4e2) at (2+.35,0+.35) {};
    \node[node_style] (v4e3) at (2-.5,0) {};
  
    \node[node_style] (v5e1) at (4-.35,2-.35) {};
    \node[node_style] (v5e2) at (4+.5,2) {}; 
    \node[node_style] (v5e3) at (4,2+.5) {};

    \node[node_style] (v6e1) at (4-.35,-2+.35) {};
    \node[node_style] (v6e2) at (4+.5,-2) {}; 
    \node[node_style] (v6e3) at (4,-2-.5) {};
    
    \node[nhidden_style] (v7) at (-6,-2) {};
    \node[nhidden_style] (v8) at (-6,2) {};
    \node[nhidden_style] (v9) at (-4,-4) {};
    \node[nhidden_style] (v10) at (-4,4) {};
    \node[nhidden_style] (v11) at (6,-2) {};
    \node[nhidden_style] (v12) at (6,2) {};
    \node[nhidden_style] (v13) at (4,-4) {};
    \node[nhidden_style] (v14) at (4,4) {};
  
    \begin{pgfonlayer}{bg} 
    \draw[edge_style]  (v1) edge (v1e1.center);
    \draw[edge_style]  (v1) edge (v1e2.center);
    \draw[edge_style]  (v1) edge (v1e3.center);

    \draw[edge_style]  (v2e1) edge (v2e2.center);
    \draw[edge_style]  (v2e2) edge (v2e3.center);
    \draw[edge_style]  (v2e3) edge (v2e1.center);

    \draw[edge_style]  (v3e1) edge (v3e2.center);
    \draw[edge_style]  (v3e2) edge (v3e3.center);
    \draw[edge_style]  (v3e3) edge (v3e1.center);

    \draw[edge_style]  (v4) edge (v4e1.center);
    \draw[edge_style]  (v4) edge (v4e2.center);
    \draw[edge_style]  (v4) edge (v4e3.center);

    \draw[edge_style]  (v5e1) edge (v5e2.center);
    \draw[edge_style]  (v5e2) edge (v5e3.center);
    \draw[edge_style]  (v5e3) edge (v5e1.center);
  
    \draw[edge_style]  (v6e1) edge (v6e2.center);
    \draw[edge_style]  (v6e2) edge (v6e3.center);
    \draw[edge_style]  (v6e3) edge (v6e1.center);
  
    \draw[edge_style]    (v1e1) edge (v3e1.center);
    \draw[edge_style]  (v2e1) edge  (v3e2.center);
   
    \draw[edge_style]  (v3e3) edge (v4e3.center);
    \draw[edge_style]  (v4e2) edge  (v5e1.center);

    \draw[edge_style]  (v4e1) edge  (v6e1.center);
    
    \draw[ehidden_style]  (v1e2) edge (v7.center);
    \draw[ehidden_style]  (v1e3) edge (v9.center);
    \draw[ehidden_style]  (v2e2) edge (v8.center); 
    \draw[ehidden_style]  (v2e3) edge (v10.center);

    \draw[ehidden_style]  (v6e2) edge (v11.center);
    \draw[ehidden_style]  (v6e3) edge (v13.center);
    \draw[ehidden_style]  (v5e2) edge (v12.center); 
    \draw[ehidden_style]  (v5e3) edge (v14.center);
    \end{pgfonlayer}
    \end{tikzpicture}
    \quad
    \begin{tikzpicture}[scale=.5,rotate=90, shorten >=1pt, auto, node distance=1cm, ultra thick]
    \tikzstyle{node_style} = [circle,draw=black,fill=black, inner sep=0pt, minimum size=3pt]
    \tikzstyle{nhidden_style} = [inner sep=0pt, minimum size=0pt]
    \tikzstyle{edge_style} = [draw=black, line width=2, thin]
    \tikzstyle{ehidden_style} = [draw=black, line width=2, dotted, thin]
    \tikzstyle{es1_style} = [draw=black, line width=2, ultra thick]
    \tikzstyle{es2_style} = [draw=gray, dashed, line width=2, ultra thick]

    \node[node_style] (v1) at (-4,-2) {};
    \node[node_style] (v1e1) at (-4+.35,-2+.35) {};
    \node[node_style] (v1e2) at (-4-.5,-2) {};
    \node[node_style] (v1e3) at (-4,-2-.5) {};
    
    \node[node_style] (v2e1) at (-4+.35,2-.35) {};
    \node[node_style] (v2e2) at (-4-.5,2) {}; 
    \node[node_style] (v2e3) at (-4,2+.5) {};
    
    \node[node_style] (v3e1) at (-2-.35,0-.35) {};
    \node[node_style] (v3e2) at (-2-.35,0+.35) {};
    \node[node_style] (v3e3) at (-2+.5,0) {};

    \node[node_style] (v4) at (2,0) {};
    \node[node_style] (v4e1) at (2+.35,0-.35) {};
    \node[node_style] (v4e2) at (2+.35,0+.35) {};
    \node[node_style] (v4e3) at (2-.5,0) {};
  
    \node[node_style] (v5e1) at (4-.35,2-.35) {};
    \node[node_style] (v5e2) at (4+.5,2) {}; 
    \node[node_style] (v5e3) at (4,2+.5) {};

    \node[node_style] (v6e1) at (4-.35,-2+.35) {};
    \node[node_style] (v6e2) at (4+.5,-2) {}; 
    \node[node_style] (v6e3) at (4,-2-.5) {};
    
    \node[nhidden_style] (v7) at (-6,-2) {};
    \node[nhidden_style] (v8) at (-6,2) {};
    \node[nhidden_style] (v9) at (-4,-4) {};
    \node[nhidden_style] (v10) at (-4,4) {};
    \node[nhidden_style] (v11) at (6,-2) {};
    \node[nhidden_style] (v12) at (6,2) {};
    \node[nhidden_style] (v13) at (4,-4) {};
    \node[nhidden_style] (v14) at (4,4) {};
  
    \begin{pgfonlayer}{bg} 
    \draw[es1_style]  (v1) edge (v1e1.center);
    \draw[edge_style]  (v1) edge (v1e2.center);
    \draw[edge_style]  (v1) edge (v1e3.center);

    \draw[edge_style]  (v2e1) edge (v2e2.center);
    \draw[edge_style]  (v2e2) edge (v2e3.center);
    \draw[edge_style]  (v2e3) edge (v2e1.center);

    \draw[edge_style]  (v3e1) edge (v3e2.center);
    \draw[edge_style]  (v3e2) edge (v3e3.center);
    \draw[es1_style]  (v3e3) edge (v3e1.center);

    \draw[edge_style]  (v4) edge (v4e1.center);
    \draw[edge_style]  (v4) edge (v4e2.center);
    \draw[es1_style]  (v4) edge (v4e3.center);

    \draw[edge_style]  (v5e1) edge (v5e2.center);
    \draw[edge_style]  (v5e2) edge (v5e3.center);
    \draw[edge_style]  (v5e3) edge (v5e1.center);
  
    \draw[edge_style]  (v6e1) edge (v6e2.center);
    \draw[edge_style]  (v6e2) edge (v6e3.center);
    \draw[edge_style]  (v6e3) edge (v6e1.center);
  
    \draw[es2_style]    (v1e1) edge (v3e1.center);
    \draw[es1_style]  (v2e1) edge [bend right] (v3e2.center);
    \draw[es2_style]  (v2e1) edge [bend left] (v3e2.center);
   
    \draw[es2_style]  (v3e3) edge (v4e3.center);
    \draw[es1_style]  (v4e2) edge [bend right] (v5e1.center);
    \draw[es2_style]  (v4e2) edge [bend left] (v5e1.center);

    \draw[es1_style]  (v4e1) edge [bend right] (v6e1.center);
    \draw[es2_style]  (v4e1) edge [bend left] (v6e1.center);
    
    \draw[ehidden_style]  (v1e2) edge (v7.center);
    \draw[ehidden_style]  (v1e3) edge (v9.center);
    \draw[ehidden_style]  (v2e2) edge (v8.center); 
    \draw[ehidden_style]  (v2e3) edge (v10.center);

    \draw[ehidden_style]  (v6e2) edge (v11.center);
    \draw[ehidden_style]  (v6e3) edge (v13.center);
    \draw[ehidden_style]  (v5e2) edge (v12.center); 
    \draw[ehidden_style]  (v5e3) edge (v14.center);
    \end{pgfonlayer}
    \end{tikzpicture}
  \]
  \caption{Left: The instance graph $G$ with two terminal nodes $z_1$ and $z_2$.
    Middle: The gadget graph $H$.
    Right: The union of two matchings in $H(X)$ and $H(Z\setminus X)$ for some $X$ containing both $z_1$ and $z_2$. Together, they form a path between $h(z_1)$ and $h(z_2)$, along with three double edges.}
    \label{fig: G'}
\end{figure}

Our algorithm is a direct application of the following result:

\begin{lemma}
\label{lem: p}
For a graph $G$, consider
\begin{equation}\label{eq: p(G,s)}
p(G,s)=\sum_{X\subseteq Z} (-1)^{|X\cap A|}\left|\sqrt{\det(D(X,s))\det(D(Z\setminus X,s))}\right|
\end{equation}
as a polynomial in the indeterminate $s$.
Let $cs^d$ be the largest degree monomial with a positive coefficient in $p(G,s)$.
Then, $\ell_{A,B}=2\Lambda-d$ is the shortest total length of any disjoint $A,B$-paths in $G$, and $c$ is $2^{|Z|/2}$ times the number of solutions having that minimum length.
\end{lemma}

\begin{proof}
  We begin by arguing that $p(G,s)$ indeed is a polynomial in the indeterminate $s$.
 Fix $X\subseteq Z$.
 Write $\mathscr M(X)$ for the set $\mathscr M(H(X))$ of perfect matchings in $H(X)$.
 By \eqref{eq: pfaff}, we can write
 \[\left|\sqrt{\det(D(X,s))}\right|
   =\operatorname{pm}(H(X))  = \sum_{M\in \mathscr M(X)} \prod_{e\in M} w(e)\,,
 \]
 where $w(e)$ is either $s^{\ell(e)}$ or $1$.
 Thus, for a pair of perfect matchings $M_1\in \mathscr M(X)$ and $M_2\in \mathscr M(Z\setminus X)$, we can write their contributing  term $t(M_1,M_2)$ as
   \[ t(M_1,M_2) = 
     \left(\prod_{e\in M_1} w(e)\right)\cdot \prod_{e\in M_2} w(e)\,,
      \]
      which is clearly a polynomial in $s$, 
      and write
      \[ p(G,s) = \sum_{X\subseteq Z} \sum_{M_1\in \mathscr M(X)}\sum_{M_2\in \mathscr M(Z\setminus X)} t(M_1,M_2)\,.
      \]
\medskip

Now view $M_1\cup M_2$ as a subgraph in $H$, by identifying each vertex in $H(X)$ and $H(Z\setminus X)$ with its copy in $H$. 
(It is helpful to view $M_1\cup M_2$ as a multiset, so the corresponding subgraph is in fact a multigraph using the edges $M_1\cap M_2$ twice.)
We can visualise this as placing the two graphs on top of each other and looking at the subgraph formed by the two matchings. 
It is clear that every vertex in $H$ has degree at most 2 in this subgraph, so $M_1\cup M_2$ can be partitioned into three edge subsets $\mathscr P, \mathscr C, \mathscr D\subseteq E(H)$, such that $\mathscr P$ is a disjoint union of simple paths, $\mathscr C$ is a disjoint union of simple cycles, and $\mathscr D$, which is equal to the intersection $M_1\cap M_2$, is a disjoint union of isolated edges.

We claim that every path in $\mathscr P$ has its endpoints in terminal centers.
To see this, first note that each terminal centre $h(z)$ for $z\in Z$ is present in exactly one of the graphs $H(X)$ and $H(Z\setminus X)$.
Therefore, $h(z)$  is matched by exactly one edge in $M_1\cup M_2$ and therefore is the endpoint of a path.
Every other vertex in $H$ appears in both $H(X)$ and $H(Z\setminus X)$ and is therefore matched in both $M_1$ and $M_2$; in particular, no such vertex is the endpoint of a simple path.
Figure~\ref{fig: G'} shows a small example.

\medskip
We next argue that unions $M_1\cup M_2$ whose paths connect terminal centers $h(a)$ and $h(b)$ with  $a\in A$ and $b\in B$ contribute nothing to $p(G,s)$. 
To this end, consider such a term $t(M_1,M_2)$ with $M_1\in \mathscr M(X)$ and $M_2\in\mathscr M(Z\setminus X)$  and let 
\[P=(u_1,\ldots,u_k) \quad\text{with }u_1=h(a), u_k=h(b)\]
be the lexicographically first such path in $\mathscr P$.

If $k$ is odd, then the edges $u_1u_2$, $u_3u_4$, $\ldots$, $u_{k-2}u_{k-1}$ belong to one matching, say $M_1$, and the edges $u_2u_3,$ $\ldots$, $u_{k-1}u_k$ belong to $M_2$.
In particular, the terminal center $h(a)$ is matched in $M_1$, which implies $h(a)\in V(H(X))$ and therefore $a\in X$.
Conversely, $h(b)$ is matched in $M_2$, which implies $h(b)\in V(H(Z\setminus X))$ and $b\notin X$.
Now form $X'=(X\cup\{b\})\setminus \{a\}$ and consider the two matchings $M_1'\in\mathscr M(X')$ and $M_2'\in\mathscr M(Z\setminus X')$ created from $M_1$ and $M_2$ by swapping the edges on $P$.
Note that the edge $u_{k-1}u_k$ incident on $h(b)$ now belongs to $M_1'$, and since $b$ belongs to $X'$, the matching $M_1'$ is indeed a perfect matching in $\mathscr M(X')$.
Similarly, $M_2'\in \mathscr M(Z\setminus X')$.
Starting the exact same process from the matchings $M_1'$ and $M_2'$ and set $X'$ would get us back to $M_1$,  $M_2$, and $X$, since the same path $P$ will be chosen by the lexicographical order, so the process defines a fixed-point free involution on the set of terms $t(M_1,M_2)$ and subsets of $Z$.

Crucially, the contribution to \eqref{eq: p(G,s)} of terms paired by this involution cancel:\[ (-1)^{|X\cap A|} t(M_1,M_2) + (-1)^{|X'\cap A|} t(M'_1,M'_2)=0\,, \]
because the multisets $M_1\cup M_2$ and $M_1'\cup M_2'$ are the same, and $X'$ and $X$ differ in exactly one terminal from $A$.
Hence no such terms will survive in the computation of $p(G,s)$.

If $k$ is even, then $u_1u_2$, $u_3u_4$, $\ldots$, $u_{k-1}u_k$ belong to the same matching, say $M_1$.
Thus, both $h(a)$ and $h(b)$ belong to $H(X)$, so $a$ and $b$ belong to $X$.
Set $X'=X\setminus \{a,b\}$, and follow the same argument as above. 

In other words, $t(M_1,M_2)$ survives in $p(G,s)$ only if the disjoint paths in $\mathscr P$ have their endpoints either both in $A$ or both in $B$. The contribution is
\[t(M_1,M_2) = \left(\prod_{e\in\mathscr D } w(e)^2 \right)\cdot\prod_{e\in \mathscr C\cup \mathscr P} w(e) = s^d\,,\]
where
\[ 
d= \left(2\sum_{e\in \mathscr D} \ell(e)\right) + \sum_{e\in \mathscr C\cup \mathscr P} \ell(e) = 2\Lambda - \sum_{e\in \mathscr C\cup \mathscr P}\ell(e)\,.
\]
The last term is at least $\ell_{A,B}$, and attains that value exactly if $\mathscr C$ is empty and $\mathscr P$ contains the external edges of a solution $E'$ to Shortest disjoint $A,B$-paths in $G$.
Otherwise, $d < 2\Lambda-\ell_{A,B}$.

\medskip
We finally turn to the other direction, to show that if there exists disjoint $A,B$-paths in $G$, we will detect them in $p(G,s)$. Moreover, we will argue that we can count the ones of shortest total length.
To see this, first consider a solution $E'\subseteq E(G)$ to Shortest disjoint $A,B$-paths, i.e., a disjoint union of paths \[E'=P_1\cup\cdots\cup P_{|Z|/2}\,,\]
each of which has terminal endpoints either both in $A$, or both in $B$.
Let $T$ be a subgraph of $H$ obtained in the following way. 
For each such path $P=(v_1,\ldots,v_k)$, first add the external edges $f(v_1v_2)g(v_1v_2)$, $\ldots$, $f(v_{k-1}v_k)g(v_{k-1}v_k)$ to $T$. Second, add the internal edges $h(v_1)f(v_1v_2)$ and $g(v_{k-1}v_k) h(v_k)$ in the two terminal gadgets, and the internal edges $g(v_iv_{i+1}) f(v_{i+1}v_{i+2})$ in the nonterminal gadgets for $i\in\{1,\ldots, k-2\}$.
This adds precisely one internal edge per gadget representing a vertex on $P$.
Third, for every vertex $u\in V(H)$ not used in an edge so far, we add to $T$ its unique external edge in $H$.
This is where we use the property of $H$ that every non-terminal vertex has a unique external edge.
Thus, $T$ consists of disjoint edge sets $\mathscr P,\mathscr D\subseteq E(H)$ where $\mathscr P$ consists of disjoint paths and $\mathscr D$ consists of disjoint (external) edges.

We continue to account for the contribution of $T$ to \eqref{eq: p(G,s)}.
Let $X\subseteq Z$ be a subset of terminals such that the endpoints of the paths in $\mathscr P$ are either both in $X$ or both in $Z\setminus X$. In particular, $|X\cap A|$ is even, and there are $2^{|Z|/2}$ such subsets.
There is exactly one perfect matching $M_1$ in $H(X)$ that is a subgraph of $T$; this matching contains all the internal edges on the paths of $\mathscr P$ with endpoints both in $X$.
There is also exactly one perfect matching $M_2$ in $H(Z\setminus X)$ that is a subgraph of $T$;
this matching contains all the external edges on the paths of $\mathscr P$ with endpoints both in $Z\setminus X$.
In particular, every external edge in $\mathscr D$ appears exactly twice in the multiset $M_1\cup M_2$, and every external edges in $\mathscr P$ appears exactly  once.
(The internal edges have weight $1$, so we need not count their contribution to a product.)
Thus, the total contribution of $M_1$ and $M_2$ is
\[ t(M_1,M_2) = \left(\prod_{e\in \mathscr D} w(e)^2\right) \cdot  \prod_{e\in \mathscr P} w(e) = s^d\,,\quad
  \text{where } d=  2\Lambda - \ell_{A,B}\,,\]
and the solution $E'$ accounts for the contribution
\[ \sum_{X\subseteq Z}(-1)^{|X\cap A|} t(M_1,M_2) = 2^{|Z|/2} s^d\,.\]
Together with the observation above that all other surviving terms have lower degree in $p(G,s)$, this shows the lemma.

\end{proof}
\subsection{Algorithm}
\label{sec: alg}
Our algorithm simply computes the coefficients of $p(G,s)$ in the definition in Lemma~\ref{lem: p} seen as a polynomial in $s$ by polynomial interpolation. The algorithm works through direct evaluation in sufficiently many points $s\in \{0,1,\ldots,2\Lambda\}$ of $p(G,s)$ after replacing $s$ for its numerical value. Hence all computations are over the integers.

\begin{enumerate}
\item For $s=0$ to $2\Lambda$,
\item \hspace{5mm} Set $sum_s=0$.
\item \hspace{5mm} For $X\subseteq Z$, $|X|$ even,
\item \hspace{10mm} Construct $H(X)$ and $H(Z\setminus X)$ and their Pfaffian orientations.
\item \hspace{10mm} Compute the integers $\det^2(D(X))$ and $\det^2(D(Z\setminus X))$ for the current value of $s$. 
\item \hspace{10mm} Take the fourth root of the two determinants and multiply them.
\item \hspace{10mm} Add the product with the sign $(-1)^{|X\cap A|}$ to $sum_s$.
\item Use polynomial interpolation to compute the coefficients of $p(G,s)$ from the array $sum$.
\item Locate the largest non-zero monomial $cs^d$.
\item Return $\ell_{A,B}=2\Lambda-d$ and $S_{A,B}=c/2^{|Z|/2}$.
\end{enumerate}

\subsection{Sequential Runtime Analysis}
We prove that the algorithm in the previous section~\ref{sec: alg} can be implemented to run sequentially in the time claimed by Thm.~\ref{thm: alg}.
Recall that the polynomial $p(G,s)$ has degree at most $2\Lambda$, and hence the number of evaluated points is sufficient to uniquely recover the coefficients of the polynomial. We can upper bound the value of the two determinants by looking at Leibniz formula for the determinant. There are at most $3^{3n+|A\cup B|}$ terms since there are at most $3$ choices per vertex in $H(X)$, For each choice the largest value is obtained if the external edges are picked twice, i.e. every term is at most $L^{2\Lambda}$. Hence the determinant can be a $\beta=\tilde{O}(nL)$ bit number.
We can compute the determinants in row 5 using $O(n^{\omega/2})$ arithmetic operations, using Yuster's algorithm~\cite{Y08} for the square of the determinant, which in turn uses the dissection method developed by Lipton et al.~\cite{LRT79}. Note that since we know all our determinants to be positive as they are squares of the number of perfect matchings, no information is lost by computing even powers of the determinant. Every arithmetic operation can be computed in $\tilde{O}(\beta)$ time~\cite{GG13}. Computing all determinants requires at most $\tilde{O}(\Lambda 2^{|A\cup B|}n^{\omega/2} \beta)=\tilde{O}(2^{|A\cup B|}n^{\omega/2+2}L^2)$ time. This part dominates the computation time, since taking the square roots in row 7 using Newton's method requires only about $\log{nL}$ iterations for an integer square, and the polynomial interpolation in row 8 can be done in quadratic time. It requires $\tilde{\Omega}(\Lambda)$ operations over a finite field, cf. \cite{GG13}, and we need a field, or several fields and the Chinese remainder theorem, of total size $\Omega(\beta)$ to recover the integer values.
This completes the proof of Thm.~\ref{thm: alg}.

\subsection{Parallel Circuit Analysis}
In this section we prove that our algorithm in section~\ref{sec: alg} can be efficiently implemented as a circuit of polynomial size and polylogarithmic depth. First we note that all values of $s$ and all values of $X$ in row 1 and 3 of the algorithm can be evaluated in parallel. All computations are made on integers of $\beta=\tilde{O}(nL)$ bits as claimed in the previous section. Addition and multiplication on $\beta$ bit integers can be done in $\operatorname{polylog}(\beta)$ depth.
Constructing the graphs $H(X)$  in row 4 can be done even without a planar embedding of $G$, as it doesn't matter how the external edges are mapped to the gadget's connectors, planarity is always preserved.
Vazirani shows that the number of perfect matchings can be computed by an NC algorithm~\cite{V89}, see also the textbook~\cite{KR98}. He describes how a Pfaffian orientation for a planar graph can be obtained via Klein and Reif's parallel planar embedding algorithm~\cite{KR86}. He next uses the fact that the determinant can be computed in NC, a consequence of Csansky's algorithm for the determinant~\cite{C76}. Berkowitz algorithm~\cite{B84} via iterated matrix product can also be used (see Cook~\cite{C85}). Computing the integer square root at row 6 is a logarithmic depth task with Newton's method since the convergence is quadratic.
Once all evaluations are done, the inner loop summation at row 7 can be computed for all $s$, again in polylogarithmic depth by a balanced binary tree of adders of $\beta$-sized integers.
Finally, Cook describes how polynomial interpolation is in NC~\cite{C85} by reducing to Berkowitz algorithm for the determinant~\cite{B84}.
This completes the proof of Thm.~\ref{thm: nc}.

\subsection{Maximum Degree 3}
\label{sec: max degree 3}

We presented our algorithm for cubic planar graphs, which is the algorithmically interesting case.
Let us observe that a simple reduction extends the algorithm to planar graphs of \emph{maximum} degree $3$, because we allow integer-weighted edges.

First consider an edge $ua$ where $u$ is a nonterminal vertex of  degree $3$ and $a$ is a terminal of degree $1$.
Then $ua$ can be removed and $u$ inserted into the terminal set of $a$; the resulting instance has a shortest solution of size $\ell_{A,B}-\ell(ua)$.
When $u$ is also a terminal vertex, there are two cases:
If $u$ belongs to the same terminal set as $a$ then $ua$ must be a path in the shortest solution, so we can remove both $u$ and $a$ and discount the resulting value by $\ell(ua)$.
If $u$ belongs to the other terminal set than $a$ then there is no solution and we can output $S_{A,B}= 0$.

Consider a $(u,v)$-path $P$ whose internal vertices all have degree $2$.
If none of $P$'s internal vertices are terminals then $P$ can be contracted into a single edge with the sum of the original edge lengths.
If $P$ contains alternating terminals, say $a\in A$, $b\in B$, $a'\in A$ in that order, no solution can exist.
If $P$ contains exactly two terminals $a\in A$ and $b\in B$ then its prefix from $u$ to $a$ can be contracted into a single edge, and so can its suffix from $b$ to $v$; the infix from $a$ to $b$ can be removed.
The resulting dangling edges $ua$ and $vb$ are handled as above.

In general, we can replace a degree-$2$ terminal $a$ incident on the edges $ua$ and $av$ with the 4-vertex `diamond' graph, introducing $3$ new nonterminal vertices.
The original edges retain their lengths, and the new edges receive length $1$, so that
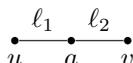
\[
  \begin{tikzpicture}[scale =.5, xscale =1.5]
    \node (u) [fill, inner sep = 1pt, circle, label = below:$u$] at (0,0) {};
    \node (a) [fill, inner sep = 1pt, circle, label = below:$a$] at (1,0) {};
    \node (v) [fill, inner sep = 1pt, circle, label = below:$v$] at (2,0) {};
    \draw (u) edge node [above] {$\ell_1$}  (a);
    \draw (a) edge node [above] {$\ell_2$} (v);
  \end{tikzpicture}
\]
 becomes
 
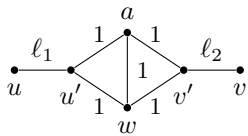
\[
   \begin{tikzpicture}[scale =.5, xscale =1.5]
     \node (u) [fill, inner sep = 1pt, circle, label = below:$u$] at (0,0) {};
     \node (x1) [fill, inner sep = 1pt, circle, label = below:$u'$] at (1,0) {};
    \node (a) [fill, inner sep = 1pt, circle, label = above:$a$] at (2,1) {};
    \node (a1) [fill, inner sep = 1pt, circle, label = below:$w$] at (2,-1) {};
    \node (y1) [fill, inner sep = 1pt, circle, label = below:$v'$] at (3,0) {};
    \node (v) [fill, inner sep = 1pt, circle, label = below:$v$] at (4,0) {};
    \draw (u) edge node [above] {$\ell_1$} (x1);
    \draw (x1) edge node [above] {\small $1$} (a);
    \draw (x1) edge node [below] {\small $1$} (a1);
    \draw (a) edge node [above] {\small $1$} (y1);
    \draw (a) edge node [right] {\small $1$} (a1);
    \draw (a1) edge node [below] {\small $1$} (y1);
    \draw (v) edge node [above] {$\ell_2$} (y1);
  \end{tikzpicture}
\]
No path with endpoint $a$ in an optimal solution will use $w$, because $uu'a$ is shorter than $uu'wa$. 
No other solution can use the nonterminal $w$ either, because doing so would isolate $a$.
We conclude that an optimal solution uses either $au'u$ or $av'v$ and no other edges in the gadget.
Thus, every optimal solution in the transformed graph corresponds to exactly one optimal solution in the original, and $\ell_{A,B}$  increments by one for each of these modifications.

\section{Hardness Result}
In this section we prove Thm.~\ref{thm: nph}.
Our hardness reduction is from counting maximum independent sets in cubic planar graphs, proven \#P-hard in Vadhan~\cite{V01} (Corollary 4.2.1). The NP-hardness result for Disjoint $A,B$-paths in general graphs by Hirai and Namba~\cite{HN17}, follows~Hirai and Pap~\cite{HP14}. It is a reduction directly from 3-Satisfiability but it is not (weakly) parsimonious.
We give here such a strengthened reduction.

Consider a cubic planar graph $G=(V,E)$ in which we want to count the maximum independent sets. We will from $G$ construct a maximum degree $3$ planar instance $I$ to Shortest disjoint $A,B$-paths. As described in the previous section, we can by adding a few vertices per vertex of degree less than three make sure the graph is cubic while preserving planarity. Here we will stick with a few vertices of degree two in our description of $I$ for simplicity.
First, for every vertex $v\in V$, we add a clockwise ordered cycle $v'_1,\ldots,v'_8$. The edge $v'_8v'_1$ has length $12$ whereas all other edges $v'_iv'_{i+1}$ for $i\in\{1,\ldots, 7\}$ have length $2$ if $i$ is odd and length $1$ if $i$ is even. Furthermore, vertices $v'_1$ and $v'_8$ belong to $A$ for every vertex $v$. Second, for every edge $uv\in E$, we add two vertices $w'_1$ and $w'_2$ to $I$. We add edges $w'_1u'_i$,$w'_2u_{i+1}$,$w'_1v'_j$, and $w'_2v_{j-1}$ of length $1$ for some indices $i$ and $j$ so that no vertex is used more than once, and the resulting graph $I$ is planar. This is easy to accomplish by using a planar embedding of $G$ and order edges incident on a vertex in clockwise order. Confer figure~\ref{fig: nph}.

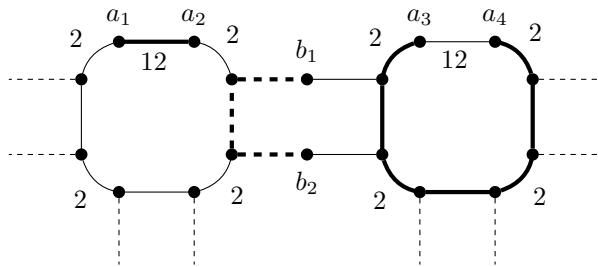
\begin{figure}[!ht]
  \centering
    \begin{tikzpicture}[scale=.5, shorten >=1pt, auto, node distance=1cm, ultra thick]
    \tikzstyle{node_style} = [circle,draw=black,fill=black, inner sep=0pt, minimum size=3pt]
    \tikzstyle{edge_style} = [draw=black, line width=2, thin]
    \tikzstyle{es1_style} = [draw=black, line width=2, ultra thick]
    \tikzstyle{es2_style} = [draw=gray, dashed, line width=2, ultra thick]
    \tikzstyle{ehidden_style} = [draw=black, line width=2, dotted, thin]
    \node[node_style] (v1) at (1,-2) {};
    \coordinate (v1e) at (1,-4) {};
    \node[node_style] (v2) at (2,-1) {};
    \node[node_style] (v3) at (2,1) {};
    \node[node_style,label=above:$a_2$] (v4) at (1,2) {};
    \node[node_style,label=above:$a_1$] (v5) at (-1,2) {};
    \node[node_style] (v6) at (-2,1) {};
    \coordinate (v6e) at (-4,1) {};
    \node[node_style] (v7) at (-2,-1) {};
    \coordinate (v7e) at (-4,-1) {};
    \node[node_style] (v8) at (-1,-2) {};
    \coordinate (v8e) at (-1,-4) {};
  
    \node[node_style,label=below:$b_2$] (v9)at (4,-1) {};
    \node[node_style,label=above:$b_1$] (v10) at (4,+1) {};
   
    \node[node_style] (v11) at (9,-2) {};
    \coordinate (v11e) at (9,-4) {};
    \node[node_style] (v12) at (10,-1) {};
    \coordinate (v12e) at (12,-1) {};
    \node[node_style] (v13) at (10,1) {};
    \coordinate (v13e) at (12,1) {};
    \node[node_style,label=above:$a_4$] (v14) at (9,2) {};
    \node[node_style,label=above:$a_3$] (v15) at (7,2) {};
    \node[node_style] (v16) at (6,1) {};
    \node[node_style] (v17) at (6,-1) {};
    \node[node_style] (v18) at (7,-2) {};
    \coordinate (v18e) at (7,-4) {};

    \begin{pgfonlayer}{bg} 
    \draw[edge_style]  (v1) edge  [bend right] node[swap] {$2$} (v2.center);
    \draw[es2_style]  (v2) edge (v3.center);
    \draw[edge_style]  (v3) edge  [bend right] node[swap] {$2$} (v4.center);
    \draw[es1_style]  (v4) edge node {$12$} (v5.center);
    \draw[edge_style]  (v5) edge  [bend right] node[swap] {$2$} (v6.center);
    \draw[edge_style]  (v6) edge (v7.center);
    \draw[edge_style]  (v7) edge [bend right] node[swap] {$2$}  (v8.center);
    \draw[edge_style]  (v8) edge (v1.center);
    \draw[ehidden_style]  (v6) edge (v6e.center); 
    \draw[ehidden_style]  (v7) edge (v7e.center); 
    \draw[ehidden_style]  (v8) edge (v8e.center); 
    \draw[ehidden_style]  (v1) edge (v1e.center); 
 
    \draw[es2_style]  (v2) edge (v9.center);
    \draw[es2_style]  (v3) edge (v10.center);
    \draw[edge_style]  (v17) edge (v9.center);
    \draw[edge_style]  (v16) edge (v10.center);

    \draw[es1_style]  (v11) edge [bend right] node[swap] {$2$}(v12.center);
    \draw[es1_style]  (v12) edge (v13.center);
    \draw[es1_style]  (v13) edge [bend right] node[swap] {$2$}(v14.center);
    \draw[edge_style]  (v14) edge node {$12$} (v15.center);
    \draw[es1_style]  (v15) edge [bend right] node[swap] {$2$}(v16.center);
    \draw[es1_style]  (v16) edge (v17.center);
    \draw[es1_style]  (v17) edge [bend right] node[swap] {$2$} (v18.center);
    \draw[es1_style]  (v18) edge (v11.center);
    \draw[ehidden_style]  (v12) edge (v12e.center); 
    \draw[ehidden_style]  (v13) edge (v13e.center); 
    \draw[ehidden_style]  (v18) edge (v18e.center); 
    \draw[ehidden_style]  (v11) edge (v11e.center); 
    \end{pgfonlayer} 
    \end{tikzpicture}
  \caption{Two vertex gadgets and one edge gadget in the constructed instance $I$ in the \#P-hardness proof. The edge terminals in $B$ must connect through some vertex gadget, forcing the $A$ terminals on it to connect through the longer length $12$ edge alternative.}
  \label{fig: nph}
\end{figure}

Furthermore, we add $w'_1$ and $w'_2$ to $B$. We now argue
\begin{lemma}
\label{lem: nph}
Let $\ell_{A,B}$ and $S_{A,B}$ be the solution to Shortest disjoint $A,B$-paths on $I$, then the maximum independent set in $G=(V,E)$ has size $\alpha(G)=12|V|+3|E|-\ell_{A,B}$ and the number of such sets are $S_{A,B}/2^{|E|-3\alpha(G)}$.
\end{lemma}
\begin{proof}
  Any vertex pair in $A$ on the same vertex gadget cycle must be connected with each other through a path, since there are no paths between different vertex gadget cycles that do not also pass through a terminal in $B$. Hence there are only two possibilities for every such pair: either it is connected through the $12$-long edge between them, or it uses the path around the cycle of length $11$. Let $I\subseteq V$ be the set of vertices whose vertex gadgets uses paths of length $11$ to connect its two $A$ terminals. The set $I$ must be an independent set in $G$, since the terminals on every edge gadget must use some edge on either of the two vertex gadgets it is connected to. Moreover, any pair of terminals in $B$ cannot be connected with a path shorter than $3$ as there exist no such short paths between any pair of them. A lower bound on the attainable length of a Shortest disjoint $A,B$-paths solution is hence $12|V|-\alpha(G)+3|E|$, where $\alpha(G)$ is the size of a maximum independent set in $G$. Any such solution can naturally be interpreted as a maximum independent set in $G$ by identifying the $A$-paths of length $11$.

Moreover, from any maximum independent set $I$ in $G$, we can construct disjoint paths of this length, simply by taking the $12$-long edge for every vertex not in $I$ for a vertex gadget's $A$ terminals, and the shorter $11$-long path for the other vertex gadgets. The edge gadgets' $B$ terminals can be connected pairwise with each other through a $3$-long path using an edge on either of its two adjacent vertex gadgets, whenever it represents a vertex not in $I$. It might be possible to connect the $B$ terminals in other ways, but those paths will be of length strictly longer than $3$ as they need to use a $2$-long edge on some vertex gadget. There are precisely $|E|-3\alpha(G)$ edges with neither endpoint in $I$, and hence the maximum independent sets will be counted $2^{|E|-3\alpha(G)}$ times in the Shortest disjoint $A,B$-paths.
\end{proof}

Thm.~\ref{thm: nph} now directly follows from Lemma~\ref{lem: nph}, since if we can find $\ell_{A,B}$ in $I$, we can also compute the number of maximum independent sets in $G$ from $S_{A,B}$.

\subparagraph*{Acknowledgements.}
We thank Radu Curticapean for making us aware of the fast algorithms for determinants of matrices with a planar structure. We are also grateful for many valuable comments by anonymous referees. 
This work was supported in part by the Swedish Research Council grant VR-2016-03855, ``Algebraic Graph Algorithms''.

\end{document}